\documentclass[11pt]{article}%
\usepackage{amsmath}
\usepackage{amsfonts}
\usepackage{amssymb}
\usepackage{graphicx}
\usepackage{fullpage}
\usepackage{hyperref}%
\providecommand{\U}[1]{\protect\rule{.1in}{.1in}}
\newtheorem{theorem}{Theorem}

\newtheorem{definition}[theorem]{Definition}

\newenvironment{proof}[1][Proof]{\noindent\textbf{#1.} }{\ \rule{0.5em}{0.5em}}
\begin{document}

\title{$\mathsf{PDQP/qpoly}=\mathsf{ALL}$}
\author{Scott Aaronson\thanks{University of Texas at Austin. \ Email:
aaronson@cs.utexas.edu. \ Supported by a Vannevar Bush Fellowship from the US
Department of Defense, a Simons Investigator Award, and the Simons
\textquotedblleft It from Qubit\textquotedblright\ collaboration.}}
\date{}
\maketitle

\begin{abstract}
We show that combining two different hypothetical enhancements to quantum
computation---namely, quantum advice and non-collapsing measurements---would
let a quantum computer solve any decision problem whatsoever in polynomial
time, even though neither enhancement yields extravagant power by itself.
\ This complements a related result due to Raz. \ The proof uses locally
decodable codes.

\end{abstract}

We've known for a quarter-century that quantum computers could efficiently
solve a few problems, like factoring and discrete logarithms, that have
resisted sustained efforts to solve them classically \cite{shor}. \ But we've
also known that the tools used to prove this don't generalize, for example, to
$\mathsf{NP}$-complete problems \cite{bbbv}. \ At least in the black-box
setting, even a quantum computer would provide at most a quadratic speedup
(i.e., the speedup of Grover's algorithm \cite{grover}) for unordered search,
and it would face similar limits for many other tasks.

This situation has motivated some researchers to consider speculative
generalizations of known physics, which would dramatically boost quantum
computers' power.\ \ In 1998, Abrams and Lloyd \cite{al}\ showed that a
nonlinear term in the Schr\"{o}dinger equation, if one existed, generally
\textit{would} let quantum computers solve $\mathsf{NP}$-complete\ and even
harder problems in polynomial time. \ Others (e.g., \cite{bacon,awat}) pointed
out similar superpowers in quantum computers equipped with closed timelike curves.

Perhaps it's no surprise that doing violence to quantum-mechanical linearity
in these ways would yield inordinate computational power. \ What's more
surprising is that there are hypothetical resources that appear to boost the
power of quantum computers, \textit{but only by a little}, rather than by
\textquotedblleft absurd\textquotedblright\ amounts. \ This note is concerned
with perhaps the two main examples of such resources: \textit{quantum advice}
and \textit{non-collapsing measurements}. \ We now discuss them in
turn.\bigskip

\textbf{Quantum Advice.} \ In 2003, Nishimura and Yamakami \cite{ny}\ defined
the class $\mathsf{BQP/qpoly}$, consisting of all decision problems solvable
by a polynomial-time quantum algorithm that's given a \textit{quantum advice
state} $\left\vert \psi_{n}\right\rangle $ with $n^{O\left(  1\right)  }$
qubits. \ The advice state depends only on the input length $n$, rather than
on the specific input $x\in\left\{  0,1\right\}  ^{n}$, but can otherwise be
chosen arbitrarily to help the algorithm. \ It's natural to wonder \textit{how
much it can help} to be given a fixed state that encodes exponentially many
complex numbers, albeit not in directly measurable form. \ More formally: does
$\mathsf{BQP/qpoly}$ equal $\mathsf{BQP/poly}$, which is the same class except
that the advice is now restricted to being classical?

Watrous \cite{watrous} gave an example of a problem for which quantum advice
seems to help. \ Given a finite group $G$, each of whose elements is uniquely
encoded by an $n$-bit string, as well as a fixed subgroup $H\leq G$ (and the
ability to perform group operations), suppose we want to decide whether an
input element $x\in G$\ belongs to $H$. \ Watrous showed that a quantum
computer can solve this problem in polynomial time, for \textit{any} $x\in G$,
if given the advice state%
\[
\left\vert H\right\rangle :=\frac{1}{\sqrt{\left\vert H\right\vert }}%
\sum_{h\in H}\left\vert h\right\rangle ,
\]
by estimating the overlap between $\left\vert H\right\rangle $\ and the coset
state $\left\vert Hx\right\rangle $ (which can be efficiently created given
$\left\vert H\right\rangle $). \ It's currently unknown how to solve the
problem without such an advice state. \ Meanwhile, Aaronson and Kuperberg
\cite{ak} showed that there exists a \textquotedblleft quantum
oracle\textquotedblright\ relative to which $\mathsf{BQP/poly}\neq
\mathsf{BQP/qpoly}$.\footnote{However, they also showed that Watrous's group
membership problem does \textit{not} lead to an oracle separation.}

Conversely, we also know significant limits on the power of quantum
advice.\ \ In 2004, Aaronson \cite{aar:adv} showed that $\mathsf{BQP/qpoly}%
\subseteq\mathsf{PostBQP/poly}$, where $\mathsf{PostBQP}$\ means quantum
polynomial-time enhanced by the ability to \textit{postselect} (or condition)
on exponentially unlikely measurement outcomes, and is known to equal the
classical complexity class $\mathsf{PP}$ \cite{aar:pp}. \ In 2010, Aaronson
and Drucker \cite{adrucker}\ improved this to $\mathsf{BQP/qpoly}%
\subseteq\mathsf{QMA/poly}\cap\mathsf{coQMA/poly}$ where $\mathsf{QMA}%
$\ (Quantum Merlin-Arthur) is a quantum analogue of $\mathsf{NP}$. \ These
results imply, by a counting argument, that there must be at least some
languages \textit{not} in $\mathsf{BQP/qpoly}$,\ which is not immediate from
the definition! \ As we'll see, there are other complexity classes
$\mathcal{C}$\ for which $\mathcal{C}\mathsf{/qpoly}$\ \textit{does} contain
all languages.

As a corollary of his so-called \textit{direct product theorem} for quantum
search, Aaronson \cite{aar:adv} also showed that there exists an oracle
relative to which $\mathsf{NP}\not \subset \mathsf{BQP/qpoly}$. \ This means
that, in the black-box setting, even quantum advice would not let quantum
computers solve $\mathsf{NP}$-complete problems in polynomial time.\bigskip

\textbf{Non-Collapsing Measurements.} \ In 2014, Aaronson et al.\ \cite{abfl}%
\ defined the class $\mathsf{PDQP}$ (Product Dynamical Quantum
Polynomial-Time), consisting of all decision problems solvable by
polynomial-time quantum algorithms with a hypothetical ability to make
\textit{multiple non-collapsing measurements} of a quantum state. \ In other
words, they considered quantum circuits that, besides $1$- and $2$-qubit
unitary gates, are equipped with two kinds of measurements:

\begin{enumerate}
\item[(i)] \textquotedblleft ordinary\textquotedblright\ measurements, which
collapse the state being measured according to the usual quantum-mechanical
rules, \textit{and also}

\item[(ii)] \textquotedblleft non-collapsing\textquotedblright\ measurements,
which return an independent sample from the appropriate output distribution
every time they're applied, yet leave the state unaffected and ready to be
measured again.
\end{enumerate}

Here Aaronson et al.\ \cite{abfl} were building on 2005 work by Aaronson
\cite{aar:qchv}, who studied the power of quantum algorithms enhanced by the
hypothetical ability to inspect the entire history of a \textit{hidden
variable} (as in Bohmian mechanics). \ This led him to define a complexity
class called $\mathsf{DQP}$ (Dynamical Quantum Polynomial-Time), which
contains $\mathsf{PDQP}$\ and is closely related to it. \ However, the later
work on $\mathsf{PDQP}$\ separated out the core complexity-theoretic issues
from the technical details of hidden-variable theories, and also fixed an
error that Aaronson \cite{aar:qchv}\ had made.\footnote{Specifically, Aaronson
\cite{aar:qchv}\ claimed to show that $\mathsf{NP}^{A}\not \subset
\mathsf{DQP}^{A}$\ relative to a suitable oracle $A$; in reality he showed no
such thing, though the conjecture remains plausible. \ By contrast, Aaronson
et al.\ \cite{abfl} gave a correct proof that $\mathsf{NP}^{A}\not \subset
\mathsf{PDQP}^{A}$\ relative to a suitable oracle $A$.}

Aaronson et al.\ \cite{abfl} gave two main examples of the power of
non-collapsing measurements. \ First, we can use non-collapsing measurements
to find \textit{collisions} in any two-to-one function $f:\left[  N\right]
\rightarrow\left[  N\right]  $---that is, pairs\ $x,y$\ such that $f\left(
x\right)  =f\left(  y\right)  $---almost instantly. \ To do so, we first
prepare the state%
\[
\frac{1}{\sqrt{N}}\sum_{x=1}^{N}\left\vert x\right\rangle \left\vert f\left(
x\right)  \right\rangle .
\]
We then apply an ordinary collapsing measurement to the second\ register, to
produce $\frac{\left\vert x\right\rangle +\left\vert y\right\rangle }{\sqrt
{2}}$\ where $f\left(  x\right)  =f\left(  y\right)  $ in the first register.
\ Finally, we apply non-collapsing measurements to the first register to read
out both $x$\ and $y$. \ Generalizing this, Aaronson et al.\ \cite{abfl}%
\ showed that $\mathsf{SZK}\subseteq\mathsf{PDQP}$, where $\mathsf{SZK}$ is
the class of problems---including, for example, graph isomorphism and breaking
lattice-based cryptography---that admit so-called statistical zero-knowledge
proof protocols.

As a second example of the power of non-collapsing measurements, Aaronson et
al. \cite{abfl} showed that they let us solve the Grover problem---i.e., given
a black-box function $f:\left[  N\right]  \rightarrow\left\{  0,1\right\}  $,
find a \textquotedblleft marked item\textquotedblright\ $x$\ such that
$f\left(  x\right)  =1$---using only $\sim N^{1/3}$\ steps, as opposed to the
$\sim\sqrt{N}$\ steps needed by an ordinary quantum computer. \ To do this, we
first run $T\sim N^{1/3}$\ iterations of Grover's search algorithm, in order
to amplify the probability of the marked item up to $\sim\frac{T^{2}}{N}%
=\frac{1}{N^{1/3}}$. \ We then make $\sim N^{1/3}$\ non-collapsing
measurements of the resulting state, until (with high probability) the marked
item has been found.

Strikingly, though, and much like with quantum advice, $\mathsf{PDQP}$\ seems
to provide only \textquotedblleft slightly\textquotedblright\ more power than
ordinary quantum computing. \ Indeed, Aaronson et al.\ \cite{abfl} showed that
any $\mathsf{PDQP}$\ algorithm needs at least $\sim N^{1/4}$\ steps to do
Grover search, and as a consequence, that there exists an oracle relative
which $\mathsf{NP}\not \subset \mathsf{PDQP}$.\ \ In other words: in the
black-box setting, even non-collapsing measurements still wouldn't let quantum
computers solve $\mathsf{NP}$-complete problems in polynomial time.\bigskip

This note considers what happens when we \textit{combine} polynomial-size
quantum advice with non-collapsing measurements, to obtain the complexity
class $\mathsf{PDQP/qpoly}$. \ Surprisingly, and contrary to our initial
guess, we find that even though the two resources are fairly weak
individually, together they let us solve\textit{ everything}. \ That is:
$\mathsf{PDQP/qpoly}=\mathsf{ALL}$, where $\mathsf{ALL}$\ is the set of all
languages $L\subseteq\left\{  0,1\right\}  ^{\ast}$ (including the halting
problem and other noncomputable languages).

There are precedents for such a result in quantum complexity theory. \ Most
notably, in 2005, Raz \cite{raz:all}\ showed that $\mathsf{QIP}\left(
2\right)  \mathsf{/qpoly}=\mathsf{ALL}$, where $\mathsf{QIP}\left(  2\right)
$\ consists of all languages that have two-message quantum interactive proof
systems. \ His protocol, though different from ours, even used the exact same
quantum advice state that ours will: namely, a superposition over a low-degree
polynomial extension of the Boolean function $f:\left\{  0,1\right\}
^{n}\rightarrow\left\{  0,1\right\}  $\ that we want to evaluate.

More trivially, Aaronson \cite{aar:qmaqpoly} observed that
$\mathsf{PostBQP/qpoly}=\mathsf{ALL}$. \ This is simply because, for any
Boolean function $f$, if given the advice state%
\begin{equation}
\frac{1}{\sqrt{2^{n}}}\sum_{z\in\left\{  0,1\right\}  ^{n}}\left\vert
z\right\rangle \left\vert f\left(  z\right)  \right\rangle \tag{*}%
\end{equation}
as well as an input $x$, we can first measure in the computational basis and
then postselect on getting $z=x$. \ For similar reasons, we have
$\mathsf{PQP/qpoly}=\mathsf{ALL}$, where $\mathsf{PQP}=\mathsf{PP}$\ consists
of all languages that admit a polynomial-time quantum algorithm that guesses
the right answer with probability greater than $1/2$. \ Using error-correcting
codes, Aaronson \cite{aar:qmaqpoly} also observed that $\mathsf{QMA}%
_{\mathsf{EXP}}\mathsf{/qpoly}=\mathsf{ALL}$,\ where $\mathsf{QMA}%
_{\mathsf{EXP}}$\ is the exponential-time analogue of $\mathsf{QMA}%
$.\footnote{Note that, as pointed out in \cite{aar:qmaqpoly}, adding quantum
advice need not \textquotedblleft commute\textquotedblright\ with standard
complexity class inclusions. \ As an example, we have $\mathsf{PP}%
=\mathsf{PostBQP}\subseteq\mathsf{BQPSPACE}=\mathsf{PSPACE}$, yet
$\mathsf{PostBQP/qpoly}$\ contains all languages whereas
$\mathsf{BQPSPACE/qpoly}=\mathsf{PSPACE/poly}$\ does not.}

Compared to these earlier observations, we think the main novelty here is
simply that $\mathsf{PDQP}$\ is so much \textit{weaker} than $\mathsf{QIP}%
\left(  2\right)  $, $\mathsf{PostBQP}$, $\mathsf{PQP}$, or $\mathsf{QMA}%
_{\mathsf{EXP}}$. \ As we've seen, unlike those other classes, $\mathsf{PDQP}%
$\ is neither known nor believed to contain $\mathsf{NP}$. \ Intuitively, it's
just a \textquotedblleft slight generalization\textquotedblright\ of
$\mathsf{BQP}$ itself---which is what makes it perhaps unsettling that the
mere addition of quantum advice can unlock so much power.

Indeed, the fact that $\mathsf{PDQP/qpoly}=\mathsf{ALL}$\ could be said to
have a \textquotedblleft real-world\textquotedblright\ implication. \ In a
forthcoming work, on a practical scheme for generating cryptographically
secure random bits using quantum supremacy experiments, Aaronson
\cite{aar:certrand}\ found that, in order to derive the soundness of such a
scheme, he needed to assume (what seems plausible) the existence of
pseudorandom functions that are indistinguishable from random functions by any
$\mathsf{PDQP}$\ algorithm. \ He then noticed that an even stronger soundness
conclusion would follow, if he assumed the existence of pseudorandom functions
that are indistinguishable from random by any $\mathsf{PDQP/qpoly}%
$\ algorithm. \ Unfortunately, by the main result of this note, the latter
doesn't exist! \ This was the genesis of the present work: as ethereal as it
sounds, the result that $\mathsf{PDQP/qpoly}=\mathsf{ALL}$\ rules out a
natural approach to proving the soundness of randomness generation schemes
against adversaries with quantum advice.\bigskip

For completeness, let us now give a formal definition of $\mathsf{PDQP/qpoly}$.

\begin{definition}
A $\mathsf{PDQP}$\ circuit, acting on $m$ qubits, is just an ordinary quantum
circuit, which starts with the initial state $\left\vert 0\right\rangle
^{\otimes m}$; and can contain $1$- and $2$-qubit unitary gates from some
finite, computationally universal set (for example, $\operatorname{CNOT}%
$\ plus $\pi/8$\ rotations), as well as measurement gates, which measure a
qubit in the $\left\{  \left\vert 0\right\rangle ,\left\vert 1\right\rangle
\right\}  $\ basis, collapsing the qubit to $\left\vert 0\right\rangle $\ or
$\left\vert 1\right\rangle $\ in the usual way. \ In a given run of the
circuit, let $\left\vert \psi_{t}\right\rangle $ be the pure state of the $m$
qubits immediately after the $t^{th}$\ gate is applied (note that the
$\left\vert \psi_{t}\right\rangle $'s can be different in different runs,
because of the probabilistic measurement gates). \ Also, let $\mathcal{D}_{t}%
$\ be the distribution over $m$-bit strings obtained by measuring $\left\vert
\psi_{t}\right\rangle $\ in the computational basis. \ Then the
\textquotedblleft output\textquotedblright\ of a $T$-gate $\mathsf{PDQP}%
$\ circuit is a list of $m$-bit strings, $y_{1},\ldots,y_{T}$, where each
$y_{t}$\ was sampled from $\mathcal{D}_{t}$,\ independently of $y_{t^{\prime}%
}$\ for all $t^{\prime}\neq t$.

A $\mathsf{PDQP}$\ algorithm is a polynomial-time classical algorithm that,
given an input $x\in\left\{  0,1\right\}  ^{n}$, gets to specify a single
$\mathsf{PDQP}$\ circuit $C=C_{x}$, receive a single output $Y=\left\langle
y_{1},\ldots,y_{T}\right\rangle $ of $C$, and finally perform classical
postprocessing on $Y$ before either accepting or rejecting.

A $\mathsf{PDQP/qpoly}$\ algorithm is the same, except that it can also
include a list of pure states $\left\{  \left\vert \psi_{n}\right\rangle
\right\}  _{n\geq1}$, where $\left\vert \psi_{n}\right\rangle $ is on
$p\left(  n\right)  $ qubits for some polynomial $p$, such that when the input
$x$\ has length $n$, the initial state of $C_{x}$\ has the form $\left\vert
\psi_{n}\right\rangle \otimes\left\vert 0\cdots0\right\rangle $\ rather than
just $\left\vert 0\right\rangle ^{\otimes m}$.

$\mathsf{PDQP/qpoly}$ is the class of languages $L\subseteq\left\{
0,1\right\}  ^{\ast}$\ for which there exists a $\mathsf{PDQP/qpoly}%
$\ algorithm $A$ such that, for all $x\in\left\{  0,1\right\}  ^{\ast}$, if
$x\in L$ then $A\left(  x\right)  $\ accepts with probability at least $2/3$,
while if $x\notin L$ then $A\left(  x\right)  $\ accepts with probability at
most $1/3$.
\end{definition}

We can also let $\mathsf{PDQEXP/qpoly}$ be the same class as
$\mathsf{PDQP/qpoly}$, except that now the quantum algorithm can use
exponential time. \ Then as an easy warmup, we observe that
$\mathsf{PDQEXP/qpoly}=\mathsf{ALL}$. \ This is simply because, given the
advice state (*), as well as an input $x\in\left\{  0,1\right\}  ^{n}$, a
$\mathsf{PDQEXP}$\ algorithm can keep measuring in the computational basis,
over and over about $2^{n}$ times, until it happens to get the outcome
$\left\vert x\right\rangle \left\vert f\left(  x\right)  \right\rangle
$.\bigskip

We now prove this note's main (only) result.

\begin{theorem}
\label{main}$\mathsf{PDQP/qpoly}=\mathsf{ALL}$.
\end{theorem}

\begin{proof}
Fix $n$, and let $f:\left\{  0,1\right\}  ^{n}\rightarrow\left\{  0,1\right\}
$ be an arbitrary Boolean function. \ Then it suffices to describe a quantum
advice state $\left\vert \psi_{f}\right\rangle $, on $n^{O\left(  1\right)  }$
qubits, such that a polynomial-time quantum algorithm equipped with both
$\left\vert \psi_{f}\right\rangle $\ and non-collapsing measurements can
evaluate $f\left(  x\right)  $ on any input $x\in\left\{  0,1\right\}  ^{n}%
$\ of its choice.

Let $\mathbb{F}$\ be a finite field of some prime order $q\geq n+2$ (by
Bertrand's postulate, we can assume $q\leq2n+1$). \ Also, let $g:\mathbb{F}%
^{n}\rightarrow\mathbb{F}$\ be the unique multilinear extension of $f$: that
is, the multilinear polynomial such that $g\left(  x\right)  =f\left(
x\right)  $\ for all $x\in\left\{  0,1\right\}  ^{n}$. \ Then our advice state
will simply be%
\[
\left\vert \psi_{f}\right\rangle :=\frac{1}{\sqrt{q^{n}}}\sum_{z\in
\mathbb{F}^{n}}\left\vert z\right\rangle \left\vert g\left(  z\right)
\right\rangle .
\]
This is a state of $O\left(  n\log n\right)  $\ qubits.

Let $R:\mathbb{F}^{n}\rightarrow\mathbb{F}^{n}$ be the function that maps each
vector $y\in\mathbb{F}^{n}$ to the unique scalar multiple $\alpha y$\ of $y$
whose leftmost nonzero entry is a $1$, or to $0^{n}$\ if $y=0^{n}$. \ In other
words, $R\left(  y\right)  $\ is a canonical label for the \textit{ray} in
$\mathbb{F}^{n}$\ that $y$\ belongs to.

Our $\mathsf{PDQP}$\ algorithm is now the following. \ Given an input
$x\in\left\{  0,1\right\}  ^{n}$, first map $\left\vert \psi_{f}\right\rangle
$\ to%
\[
\frac{1}{\sqrt{q^{n}}}\sum_{z\in\mathbb{F}^{n}}\left\vert z\right\rangle
\left\vert g\left(  z\right)  \right\rangle \left\vert R\left(  z-x\right)
\right\rangle .
\]
Then measure the third\ register, $\left\vert R\left(  z-x\right)
\right\rangle $, via an ordinary collapsing measurement.

If the measurement outcome happens to be $0^{n}$, then we can immediately
learn $g\left(  x\right)  =f\left(  x\right)  $\ by simply measuring the
second register.

In the much more likely case that measuring $\left\vert R\left(  z-x\right)
\right\rangle $\ yielded a nonzero outcome, say $y$, the reduced state of the
first two registers is now%
\[
\left\vert \phi\right\rangle :=\frac{1}{\sqrt{q-1}}\sum_{j\in\mathbb{F}%
\setminus\left\{  0\right\}  }\left\vert x+jy\right\rangle \left\vert g\left(
x+jy\right)  \right\rangle .
\]
Define $p:\mathbb{F}\rightarrow\mathbb{F}$\ by $p\left(  j\right)  :=g\left(
x+jy\right)  $. \ Then notice that $p$ is a univariate polynomial in $j$ of
degree at most $n$, and furthermore that $p\left(  0\right)  =g\left(
x\right)  =f\left(  x\right)  $.

As the last step, we simply perform repeated non-collapsing measurements of
$\left\vert \phi\right\rangle $\ in the computational basis, until we have
learned the values of $p\left(  j\right)  $\ for every $j\in\mathbb{F}%
\setminus\left\{  0\right\}  $. \ This is an instance of the coupon
collector's problem, so with overwhelming probability it takes at most
$O\left(  q\log q\right)  =O\left(  n\log n\right)  $\ measurements. \ Then,
in the classical postprocessing phase, we perform polynomial interpolation on
the recovered $p\left(  j\right)  $\ values, in order to learn $p\left(
0\right)  =f\left(  x\right)  $.
\end{proof}

\bigskip We conclude with some miscellaneous remarks and open problems about
Theorem \ref{main}.

Notice that the proof of Theorem \ref{main} did not depend on quantum
mechanics in any essential way. \ In other words, let $\mathsf{PDPP}$\ be a
classical analogue of $\mathsf{PDQP}$, in which we can execute a
polynomial-time randomized algorithm, while performing both \textquotedblleft
collapsing\textquotedblright\ and \textquotedblleft
non-collapsing\textquotedblright\ measurements of the algorithm's current
probabilistic state. \ Also, let $\mathsf{PDPP/rpoly}$\ be $\mathsf{PDPP}%
$\ augmented with polynomial-size randomized advice. \ Then exactly the same
argument gives us%
\[
\mathsf{PDPP/rpoly}=\mathsf{ALL}.
\]
The previous results of Raz \cite{raz:all}\ and Aaronson \cite{aar:qmaqpoly},
about quantum advice boosting various quantum complexity classes to unlimited
power, can all similarly be \textquotedblleft de-quantized,\textquotedblright%
\ and stated in terms of randomized rather than quantum advice. \ That is,%
\[
\mathsf{IP}\left(  2\right)  \mathsf{/rpoly}=\mathsf{PostBPP/rpoly}%
=\mathsf{PP/rpoly}=\mathsf{MA}_{\mathsf{EXP}}\mathsf{/rpoly}=\mathsf{ALL}%
.\footnote{A word of caution, though: even though Goldwasser and Sipser
\cite{gs} showed that $\mathsf{IP}\left(  2\right)  =\mathsf{AM}$ (where
$\mathsf{AM}$\ denotes two-message, \textit{public-coin} interactive proof
systems), we do \textit{not} have $\mathsf{AM/rpoly}=\mathsf{ALL}$. \ Instead,
$\mathsf{AM/rpoly}=\mathsf{MA/rpoly}=\mathsf{NP/poly}$\ (see
\cite{aar:qmaqpoly}). \ This is an instance of the broader phenomenon that
adding randomized and quantum advice needn't commute with standard complexity
class containments.}%
\]
Indeed, the only reason to state these results in terms of quantum advice in
the first place, is that quantum advice has been a subject of independent
interest whereas randomized advice has not.

In 2006, Aaronson \cite{aar:qmaqpoly} raised the question of whether there's
\textit{any} natural quantum complexity class $\mathcal{C}$\ that quantum
advice boosts to $\mathsf{ALL}$, even though classical randomized advice fails
to do so. \ As far as we know that question remains open.\bigskip

The trick used to prove Theorem \ref{main} also has an implication for
communication complexity. \ Namely: suppose Alice has a string $x\in\left\{
0,1\right\}  ^{N}$, Bob has an index $i\in\left[  N\right]  $, and Alice wants
to send Bob a message that will enable him to learn $x_{i}$. \ For this
so-called Index problem, it's known that even any quantum protocol requires
Alice\ to send Bob at least $\sim N$\ qubits \cite{antv}. \ Nevertheless, we
claim that there's a protocol for this problem in which Alice sends Bob a
quantum state $\left\vert \psi_{x}\right\rangle $\ of only $O\left(  \log
N\log\log N\right)  $\ qubits, and then Bob learns $x_{i}$\ after making an
ordinary collapsing measurement of $\left\vert \psi_{x}\right\rangle
$\ followed by $O\left(  \log N\log\log N\right)  $\ non-collapsing
measurements. \ This protocol is exactly the one from Theorem \ref{main},
except with $x$ in place of the truth table of $f$, and $i$ in place of
$x$.\bigskip

Any reader familiar with \textit{Locally Decodable Codes} (LDCs) might
recognize them as the central concept in the proof of Theorem \ref{main}, even
though we kept the proof self-contained and never used the term. \ In general,
an error-correcting code is a function $C:\Sigma^{N}\rightarrow\Sigma^{M}$ for
some finite alphabet $\Sigma$, with the property that $C\left(  x\right)
$\ and $C\left(  y\right)  $\ differ on a large fraction of coordinates for
all $x\neq y$. \ An LDC is a special kind of error-correcting code: one such
that, for each entry $x_{i}$\ of the original string $x=x_{1}\ldots x_{N}$,
it's possible to recover $x_{i}$\ from any string $w$\ close to $C\left(
x\right)  $, with high probability, via a randomized algorithm that queries
$w$\ in only $r$\ randomly chosen locations. \ Here one wants $r$ to be as
small as possible, even a constant like $2$ or $3$.

In a sequence of breakthroughs (see, e.g.,
\cite{yekhanin,itohsuzuki,efremenko}), it was established that for every
constant $r=2^{t}$, there exist $r$-query LDCs with linear distance and with
size%
\[
M=\exp\exp\left(  \left(  \log N\right)  ^{1/t}\left(  \log\log N\right)
^{1-1/t}\right)  .
\]
For $r\geq4$,\footnote{Though $4$ is the smallest power of $2$ for which the
bound is nontrivial, with modified arguments one can also handle the case
$r=3$.} this size is less than exponential in $N$, albeit more than
polynomial. \ We didn't use these sophisticated LDCs, for a combination of
reasons: first, we were fine with $r=n^{O\left(  1\right)  }$\ queries, which
meant that a vastly simpler LDC, based on a multilinear extension of the
Boolean function $f$, could be used instead. \ Second, we were \textit{not}
fine with $\log M$, the number of qubits in the advice state, being more than
$\left(  \log N\right)  ^{O\left(  1\right)  }=n^{O\left(  1\right)  }$, as it
would be with the state-of-the-art constant-query LDCs.

One might ask whether, in the algorithm of Theorem \ref{main}, the number of
non-collapsing measurements could be reduced from $O\left(  n\log n\right)
$\ to a small constant $r$. \ A positive answer will follow if there turn out
to exist $\left(  r+1\right)  $-query LDCs of constant distance and at most
quasipolynomial size, which moreover are sufficiently explicit and efficient.

In this connection, it's interesting that Kerenidis and de Wolf \cite{kw}%
\ proved---as it happens, by using a quantum information argument---that there
are no $2$-query LDCs of subexponential size. \ This raises the possibility
that, in any algorithm like ours, there must be at least \textit{two}
non-collapsing measurements (as well as a third and final measurement, which
might as well be collapsing). \ This seems surprising: \textit{a priori}, one
might have guessed that a single non-collapsing measurement would already
provide all the computational power that can be had from such a
resource.\bigskip

The open problem that interests us the most in this subject is the following.
\ A central fact about $\mathsf{PDQP}$, shown by Aaronson et al.\ \cite{abfl},
is that it contains $\mathsf{SZK}$. \ While \cite{abfl}\ never made this
explicit, the same argument shows that $\mathsf{PDQP}$\ contains a larger
class that we could call $\mathsf{QCSZK}$ (Quantum Classical $\mathsf{SZK}$),
consisting of all languages that admit a statistical zero-knowledge proof
protocol with a quantum verifier but classical communication with the
prover.\footnote{This class has the following as a complete promise problem,
generalizing the $\mathsf{SZK}$-complete Statistical Difference problem of
Sahai and Vadhan \cite{sv}. \ Given as input two quantum circuits $C_{0}$\ and
$C_{1}$, which sample probability distributions $\mathcal{D}_{0}$\ and
$\mathcal{D}_{1}$ respectively over $n$-bit strings, decide whether
$\mathcal{D}_{0}$\ and $\mathcal{D}_{1}$\ have variation distance at most
$1/3$\ or at least $2/3$, promised that one of these is the case.} \ We thus
raise the following question: does $\mathsf{QCSZK/qpoly}$\ equal
$\mathsf{ALL}$? \ Or we might as well ask the analogous classical question:
does $\mathsf{SZK/rpoly}$\ equal $\mathsf{ALL}$? \ What about
$\mathsf{NISZK/rpoly}$\ (where $\mathsf{NISZK}$\ means Non-Interactive
$\mathsf{SZK}$)?

\section{Acknowledgments}

I thank Dana Moshkovitz for helpful conversations.

\bibliographystyle{plain}
\bibliography{thesis}

\begin{thebibliography}{10}

\bibitem{aar:adv}
S.~Aaronson.
\newblock Limitations of quantum advice and one-way communication.
\newblock {\em Theory of Computing}, 1:1--28, 2005.
\newblock Earlier version in CCC'2004. quant-ph/0402095.

\bibitem{aar:qchv}
S.~Aaronson.
\newblock Quantum computing and hidden variables.
\newblock {\em Phys. Rev. A}, 71(032325), 2005.
\newblock quant-ph/0408035 and quant-ph/0408119.

\bibitem{aar:pp}
S.~Aaronson.
\newblock Quantum computing, postselection, and probabilistic polynomial-time.
\newblock {\em Proc. Roy. Soc. London}, A461(2063):3473--3482, 2005.
\newblock quant-ph/0412187.

\bibitem{aar:qmaqpoly}
S.~Aaronson.
\newblock {QMA/qpoly} is contained in {PSPACE/poly}: de-{M}erlinizing quantum
  protocols.
\newblock In {\em Proc. Conference on Computational Complexity}, pages
  261--273, 2006.
\newblock quant-ph/0510230.

\bibitem{aar:certrand}
S.~Aaronson.
\newblock Certified randomness from quantum supremacy.
\newblock To appear, 2018.

\bibitem{abfl}
S.~Aaronson, A.~Bouland, J.~Fitzsimons, and M.~Lee.
\newblock The space ``just above'' {BQP}.
\newblock In {\em Proc. Innovations in Theoretical Computer Science (ITCS)},
  pages 271--280, 2016.
\newblock arXiv:1412.6507.

\bibitem{adrucker}
S.~Aaronson and A.~Drucker.
\newblock A full characterization of quantum advice.
\newblock {\em SIAM J. Comput.}, 43(3):1131--1183, 2014.
\newblock Earlier version in STOC'2010. arXiv:1004.0377.

\bibitem{ak}
S.~Aaronson and G.~Kuperberg.
\newblock Quantum versus classical proofs and advice.
\newblock {\em Theory of Computing}, 3(7):129--157, 2007.
\newblock Earlier version in CCC'2007. arXiv:quant-ph/0604056.

\bibitem{awat}
S.~Aaronson and J.~Watrous.
\newblock Closed timelike curves make quantum and classical computing
  equivalent.
\newblock {\em Proc. Roy. Soc. London}, (A465):631--647, 2009.
\newblock arXiv:0808.2669.

\bibitem{al}
D.~S. Abrams and S.~Lloyd.
\newblock Nonlinear quantum mechanics implies polynomial-time solution for
  {NP}-complete and {\#}{P} problems.
\newblock {\em Phys. Rev. Lett.}, 81:3992--3995, 1998.
\newblock quant-ph/9801041.

\bibitem{antv}
A.~Ambainis, A.~Nayak, A.~Ta-Shma, and U.~V. Vazirani.
\newblock Quantum dense coding and quantum finite automata.
\newblock {\em J. of the ACM}, 49:496--511, 2002.
\newblock Earlier version in STOC'1999, pp. 376-383. quant-ph/9804043.

\bibitem{bacon}
D.~Bacon.
\newblock Quantum computational complexity in the presence of closed timelike
  curves.
\newblock {\em Phys. Rev. A}, 70(032309), 2004.
\newblock quant-ph/0309189.

\bibitem{bbbv}
C.~Bennett, E.~Bernstein, G.~Brassard, and U.~Vazirani.
\newblock Strengths and weaknesses of quantum computing.
\newblock {\em SIAM J. Comput.}, 26(5):1510--1523, 1997.
\newblock quant-ph/9701001.

\bibitem{efremenko}
K.~Efremenko.
\newblock 3-query locally decodable codes of subexponential length.
\newblock {\em SIAM J. Comput.}, 41(6):1694--1703, 2012.
\newblock Earlier version in STOC'2009. ECCC TR08-069.

\bibitem{gs}
S.~Goldwasser and M.~Sipser.
\newblock Private coins versus public coins in interactive proof systems.
\newblock In {\em Randomness and Computation}, volume~5 of {\em Advances in
  Computing Research}. JAI Press, 1989.

\bibitem{grover}
L.~K. Grover.
\newblock A fast quantum mechanical algorithm for database search.
\newblock In {\em Proc. ACM STOC}, pages 212--219, 1996.
\newblock quant-ph/9605043.

\bibitem{itohsuzuki}
T.~Itoh and Y.~Suzuki.
\newblock New constructions for query-efficient locally decodable codes of
  subexponential length.
\newblock arXiv:0810.4576, 2008.

\bibitem{kw}
I.~Kerenidis and R.~de Wolf.
\newblock Exponential lower bound for 2-query locally decodable codes via a
  quantum argument.
\newblock {\em J. Comput. Sys. Sci.}, 69(3):395--420, 2004.
\newblock Earlier version in STOC'2003. quant-ph/0208062.

\bibitem{ny}
H.~Nishimura and T.~Yamakami.
\newblock Polynomial time quantum computation with advice.
\newblock {\em Inform. Proc. Lett.}, 90:195--204, 2003.
\newblock ECCC TR03-059, quant-ph/0305100.

\bibitem{raz:all}
R.~Raz.
\newblock Quantum information and the {PCP} theorem.
\newblock In {\em Proc. IEEE FOCS}, pages 459--468, 2005.
\newblock quant-ph/0504075.

\bibitem{sv}
A.~Sahai and S.~Vadhan.
\newblock A complete promise problem for statistical zero-knowledge.
\newblock {\em J. of the ACM}, 50(2):196--249, 2003.
\newblock Earlier version in FOCS'1997. ECCC TR00-084.

\bibitem{shor}
P.~W. Shor.
\newblock Polynomial-time algorithms for prime factorization and discrete
  logarithms on a quantum computer.
\newblock {\em SIAM J. Comput.}, 26(5):1484--1509, 1997.
\newblock Earlier version in FOCS'1994. quant-ph/9508027.

\bibitem{watrous}
J.~Watrous.
\newblock Succinct quantum proofs for properties of finite groups.
\newblock In {\em Proc. IEEE FOCS}, pages 537--546, 2000.
\newblock cs.CC/0009002.

\bibitem{yekhanin}
S.~Yekhanin.
\newblock Towards 3-query locally decodable codes of subexponential length.
\newblock {\em J. of the ACM}, (55):1, 2008.
\newblock Earlier version in STOC'2007. See also ECCC TR06-127.

\end{thebibliography}

\end{document}